\documentclass[aps,preprint,nofootinbib,superscriptaddress,longbibliography,showkeys]{revtex4-2}

\usepackage[T1]{fontenc}
\usepackage{lmodern}
\usepackage{amsmath,amssymb,bm,amsthm}
\usepackage{graphicx}
\usepackage{booktabs}
\usepackage{microtype}
\usepackage{url}
\usepackage[hidelinks]{hyperref}
\hypersetup{
  pdftitle={Einstein--Langevin source correlations in closed non-relativistic systems},
  pdfauthor={Peng Xu and Li-e Qiang},
  pdfkeywords={Einstein--Langevin equation, stochastic semiclassical gravity, stress--energy fluctuations, unequal-time source correlations, Diósi--Penrose model, gravitational decoherence}
}

\newtheorem{proposition}{Proposition}
\begin{document}

\title{Einstein--Langevin source correlations in closed non-relativistic systems}

\author{Peng Xu}
\email{xupeng@imech.ac.cn}
\affiliation{Center for Gravitational Wave Experiment, National Microgravity Laboratory, Institute of Mechanics, Chinese Academy of Sciences, Beijing 100190, China}
\affiliation{Key Laboratory of Gravitational Wave Precision Measurement of Zhejiang Province, Hangzhou Institute for Advanced Study, UCAS, Hangzhou 310024, China}
\affiliation{Taiji Laboratory for Gravitational Wave Universe (Beijing/Hangzhou), University of Chinese Academy of Sciences (UCAS), Beijing 100049, China}

\author{Li-e Qiang}
\email{qianglie@nssc.ac.cn}
\affiliation{National Space Science Center, Chinese Academy of Sciences, Beijing 100190, China}

\begin{abstract}
Stochastic semiclassical gravity represents stress--energy fluctuations through the Einstein--Langevin noise kernel. We examine its Newtonian projection for specified closed, particle-number-conserving non-relativistic sources and distinguish the source covariance from a reduced matter channel. The connected one-particle covariance is state dependent and differs from the universal potential-noise law postulated in the Di\'osi--Penrose model. To relate the source spectrum to dephasing, we define a separate static probe subject to a prescribed Gaussian noise law. For a finite-variance source operator in an isolated non-degenerate eigenstate, the full connected spectral measure has a gap around zero. The probe dephasing exponent is then bounded and its long-time slope vanishes. Stationary doublets can instead carry static or low-frequency spectral lines, while a specified freely spreading Gaussian source gives a finite limiting exponent. Continuum results depend on preparation, normalisation, infrared regularisation and the order of limits. We also compare coherent and incoherent many-body preparations and establish no-signalling for fixed classical-noise channels. These results identify the spatial, temporal and dynamical assumptions needed to relate source fluctuations to a matter-decoherence model. They do not exclude Markovian limits in other number-conserving non-relativistic preparations.
\end{abstract}


\maketitle
\clearpage

\section{Introduction}\label{sec:intro}

Semiclassical gravity couples a classical spacetime geometry to the expectation value of a quantum stress--energy tensor. Its weak-field non-relativistic limit gives the Schr\"odinger--Newton equation~\cite{Moller1962,Rosenfeld1963,Kiefer1994,Diosi1984,Penrose1996,Giulini2012,Moroz1998}. The equation is widely used as a mean-field model of self-gravitating matter~\cite{Giulini2014}, although its status as a single-particle equation derived from quantum field theory has been questioned~\cite{Anastopoulos2014}. Related applications include ultralight-matter cosmology~\cite{Hu2000,Guzman2004,Kopp2017} and optical analogues~\cite{Bekenstein2015,Roger2016,Paz2025}. The Schr\"odinger--Newton equation is not an open-system equation. By itself it neither produces decoherence nor resolves the causality problems that arise when deterministic nonlinear evolution is combined with the projection postulate~\cite{Bahrami2014,Bielinska2025}.

The Di\'osi--Penrose (DP) model instead postulates temporally white Newtonian potential noise with inverse-distance spatial correlations and a short-distance smearing scale~\cite{Diosi2007,Bassi2003,Bassi2013,Bassi2017}. Conventions differ by factors of two depending on how the noise kernel and Lindblad coefficient are normalised. These differences do not affect the structural comparison below. Collapse models have motivated experimental tests of gravity-related decoherence~\cite{Nimmrichter2014,Carlesso2022,Donadi2021,Vinante2020,Fadel2023} and broader theoretical analyses~\cite{Tilloy2019,Asprea2021,Anastopoulos2022}. Consistent Newtonian hybrid dynamics can also be constructed through continuous measurement and feedback, in which case the decoherence kernel is part of the specified hybrid coupling~\cite{Kafri2014,Tilloy2016,Tilloy2017}. Related first-principles master-equation derivations based on weak gravitational correlations were given by Anastopoulos and Hu~\cite{Anastopoulos2013}. Bera, Mohan and Singh already obtained the connected single-particle gravitational-potential covariance, nonwhite correlations and an unequal-time perturbative master-equation treatment within the Einstein--Langevin (EL) framework~\cite{Bera2015}. The present paper examines when these preparation-dependent spectra yield a long-time dephasing slope for a specified static probe.

The issue is a low-energy matching problem in stochastic semiclassical gravity. The EL kernel is a covariant connected correlator of the matter stress--energy tensor and provides a stochastic representation of metric fluctuations. A matter-decoherence channel is a further reduced description. Neither temporal whiteness nor a dephasing rate follows from an equal-time covariance. For a probe with a static diagonal coupling, a regular continuous spectrum near zero can determine the long-time slope under suitable hypotheses. Transition-inducing couplings can instead sample nonzero Bohr frequencies. A comparison with DP noise must therefore specify the source covariance, gravitational response and probe dynamics separately.

Giovannini's analysis of relic gravitons also illustrates the importance of retaining correlations that depend on the two times separately~\cite{Giovannini2024}. That radiative cosmological setting differs from the matter sources studied here, and nonstationarity alone does not establish non-Markovian reduced dynamics.

We carry out this comparison for closed sources evolving autonomously in the particle-number-conserving non-relativistic limit, with a distinct static probe used to diagnose their correlations. We calculate the one-particle source covariance and its Newtonian projection, then examine isolated bound states, stationary doublets, a normalisable free Gaussian packet and continuum preparations. We also derive the fixed-number many-body covariance and establish no-signalling for prescribed classical-noise channels. The analysis distinguishes source spectral properties, comparison with the DP law and construction of a reduced matter channel. It does not exclude Markovian limits for other preparations.

\section{Einstein--Langevin kernel and strict non-relativistic reduction}\label{sec:covariant}

In stochastic semiclassical gravity, linear fluctuations about a semiclassical background are represented by a stochastic source $\xi_{\mu\nu}$~\cite{Calzetta2008,Hu2004,Martin1999,Phillips2001}. For a metric perturbation $h_{\mu\nu}$,
\begin{equation}
\mathcal G^{(1)}_{\mu\nu}[h]-\frac{8\pi G}{c^4}\langle \hat T^{(1)}_{\mu\nu}[h]\rangle
=\frac{8\pi G}{c^4}\,\xi_{\mu\nu},
\label{eq:ELcov}
\end{equation}
where $\mathcal G^{(1)}_{\mu\nu}$ is the linearised gravitational operator. The source has zero mean and covariance
\begin{align}
\langle \xi_{\mu\nu}(x)\xi_{\alpha\beta}(x')\rangle_{\rm s}
&=N_{\mu\nu\alpha\beta}(x,x')\nonumber\\
&=\frac12\langle\{\hat t_{\mu\nu}(x),\hat t_{\alpha\beta}(x')\}\rangle,
\label{eq:noise_covariant}
\end{align}
with $\hat t_{\mu\nu}=\hat T_{\mu\nu}-\langle\hat T_{\mu\nu}\rangle$. This stochastic representation reproduces symmetrised metric correlations obtained from the closed-time-path influence functional. It is not, by itself, a reduced equation for the matter state. Treating $\xi_{\mu\nu}$ as a bath acting back on that state requires an additional coarse graining and a specified temporal spectrum.

For weak fields, write $g_{00}=-1-2\Phi/c^2+O(c^{-4})$. The $00$ component of Eq.~\eqref{eq:ELcov} reduces to a stochastic Poisson equation with $\nabla^2\xi_\Phi=4\pi G\,\xi_{00}/c^2$. Since $\nabla^2|\mathbf r-\mathbf x|^{-1}=-4\pi\delta^{(3)}(\mathbf r-\mathbf x)$, the fluctuating Newtonian potential is
\begin{equation}
\xi_\Phi(\mathbf r,t)=-G\int d^3x\,\frac{\xi_{00}(\mathbf x,t)/c^2}{|\mathbf r-\mathbf x|}.
\label{eq:xi_projection}
\end{equation}
The sign drops out of the potential covariance,
\begin{align}
C_\Phi(\mathbf r,t;\mathbf r',t')
&=\langle \xi_\Phi(\mathbf r,t)\xi_\Phi(\mathbf r',t')\rangle_{\rm s}\nonumber\\
&=G^2\int d^3x\,d^3y\,
\frac{N_{0000}(\mathbf x,t;\mathbf y,t')/c^4}
{|\mathbf r-\mathbf x|\,|\mathbf r'-\mathbf y|}.
\label{eq:Cphi_cov}
\end{align}
Equation~\eqref{eq:Cphi_cov} keeps the matter source covariance distinct from the Newtonian response.

For a relativistic complex scalar field, the particle and antiparticle components may be written at leading non-relativistic order as
\begin{equation}
\hat\Phi\simeq\sqrt{\frac{\hbar}{2m_sc^2}}
\left[e^{-im_sc^2t/\hbar}\hat\psi_s+e^{+im_sc^2t/\hbar}\hat\chi_s^\dagger\right].
\label{eq:NRdecomp}
\end{equation}
The slow fields $\hat\psi_s$ and $\hat\chi_s$ annihilate independent particles and antiparticles. At leading order, they obey Schr\"odinger evolution. The adjoint in the second term is required because the complex field creates an antiparticle. Momentum-amplitude corrections are of relative order $v^2/c^2$.

We reduce the stress operator before forming its covariance. If $P$ projects onto the no-antiparticle, fixed-particle-number sector, the retained mass density is
\begin{equation}
\hat\mu_s=\left.\frac{P:\hat T_{00}:P}{c^2}\right|_{\text{leading non-relativistic order}}
=m_sP\hat\psi_s^\dagger\hat\psi_sP.
\label{eq:T00NR}
\end{equation}
Normal ordering removes the vacuum one-point term, whereas projection removes pair channels from the retained operator. These operations are distinct: for a Hermitian smeared stress operator $A$,
\begin{equation}
PA^2P=(PAP)^2+PA(1-P)AP.
\label{eq:projection_before_covariance}
\end{equation}
Restricting the external state alone would retain the second term. The retained covariance is formed after projection and therefore excludes the omitted relativistic pair and vacuum channels.

Below, $N^{\rm ret}_{0000}/c^4$ denotes the centred symmetrised covariance of $\hat\mu_s$, with ``ret'' meaning retained, not retarded. Its Newtonian projection is $C_\Phi$ in Eq.~\eqref{eq:Cphi_cov}. The rest-energy phase cancels from the retained density bilinear, so its spectrum is governed by energy differences of the source Hamiltonian. This reduction also omits subleading stress--energy corrections, radiative metric modes and retarded relativistic backreaction. It does not determine the low-frequency weight for every source state or define a reduced matter channel.

\section{Fixed-number equal-time source covariance}\label{sec:equal_time}

Writing $\hat\psi$ for the retained source field, we consider a non-relativistic bosonic field obeying $[\hat\psi(\mathbf r),\hat\psi^\dagger(\mathbf r')]=\delta^{(3)}(\mathbf r-\mathbf r')$ and the normalised one-particle state
\begin{equation}
|\Psi\rangle=\int d^3z\,\psi(\mathbf z)\hat\psi^\dagger(\mathbf z)|0\rangle,
\qquad \int d^3z\,|\psi(\mathbf z)|^2=1.
\label{eq:oneparticle}
\end{equation}
With $p(\mathbf x)=|\psi(\mathbf x)|^2$, the equal-time $00$-$00$ kernel is
\begin{align}
\frac{N^{\rm ret}_{0000}(\mathbf x,\mathbf y)}{c^4}
&=m_s^2K_1(\mathbf x,\mathbf y),\nonumber\\
K_1(\mathbf x,\mathbf y)
&=\frac12\langle\{\delta\hat n(\mathbf x),\delta\hat n(\mathbf y)\}\rangle.
\label{eq:Kdef_CQG}
\end{align}
The density product satisfies
\begin{equation}
\hat n(\mathbf x)\hat n(\mathbf y)
=\delta^{(3)}(\mathbf x-\mathbf y)\hat\psi^\dagger(\mathbf x)\hat\psi(\mathbf y)
+\hat\psi^\dagger(\mathbf x)\hat\psi^\dagger(\mathbf y)\hat\psi(\mathbf x)\hat\psi(\mathbf y).
\label{eq:nn_product}
\end{equation}
The second term annihilates the one-particle state. Consequently,
\begin{equation}
\langle\hat n(\mathbf x)\hat n(\mathbf y)\rangle
=p(\mathbf x)\delta^{(3)}(\mathbf x-\mathbf y),
\end{equation}
and the connected source covariance is
\begin{equation}
K_1(\mathbf x,\mathbf y)
=p(\mathbf x)\delta^{(3)}(\mathbf x-\mathbf y)-p(\mathbf x)p(\mathbf y).
\label{eq:K1_CQG}
\end{equation}
Since $\hat n(\mathbf x)|0\rangle=0$ in the reduced non-relativistic theory, no additional vacuum subtraction is required in this number-density sector.

The connected subtraction enforces
\begin{equation}
\int d^3y\,K_1(\mathbf x,\mathbf y)=0,\qquad\int d^3x\,d^3y\,K_1(\mathbf x,\mathbf y)=0,
\label{eq:masscons}
\end{equation}
so the total mass of a fixed-particle-number state does not fluctuate. As a connected covariance, $K_1$ is positive semidefinite. Explicitly, for every real test function $f$ with finite second moment,
\begin{align}
&\int d^3x\,d^3y\,f(\mathbf x)K_1(\mathbf x,\mathbf y)f(\mathbf y)\nonumber\\
&\qquad=\int d^3x\,p(\mathbf x)f^2(\mathbf x)
-\left[\int d^3x\,p(\mathbf x)f(\mathbf x)\right]^2\ge0.
\label{eq:positivity}
\end{align}

Substituting Eq.~\eqref{eq:K1_CQG} into Eq.~\eqref{eq:Cphi_cov} at equal time gives
\begin{equation}
C_\Phi(\mathbf r,\mathbf r')
=G^2m_s^2\left[\mathcal I(\mathbf r,\mathbf r')-\phi(\mathbf r)\phi(\mathbf r')\right],
\label{eq:Cphi_final_CQG}
\end{equation}
where
\begin{align}
\mathcal I(\mathbf r,\mathbf r')
&=\int d^3x\,\frac{p(\mathbf x)}{|\mathbf r-\mathbf x|\,|\mathbf r'-\mathbf x|},\nonumber\\
\phi(\mathbf r)&=\int d^3x\,\frac{p(\mathbf x)}{|\mathbf r-\mathbf x|}.
\label{eq:Iphi}
\end{align}
For a localised source, the kernel generally depends on the wave packet and not only on $\mathbf r-\mathbf r'$. It is therefore not translation invariant. The connected spatial structure was obtained in Ref.~\cite{Bera2015}. Our prefactor, fixed by Eqs.~\eqref{eq:noise_covariant} and \eqref{eq:xi_projection}, is four times that in its Eq.~(22).

For the two-point Fourier transform we use
\begin{equation}
\widetilde K_1(\mathbf k,\mathbf k')
=\int d^3x\,d^3y\,e^{-i\mathbf k\cdot\mathbf x}e^{+i\mathbf k'\cdot\mathbf y}K_1(\mathbf x,\mathbf y),
\label{eq:two_point_FT}
\end{equation}
with $\widetilde f(\mathbf k)=\int d^3x\,e^{-i\mathbf k\cdot\mathbf x}f(\mathbf x)$. Then
\begin{equation}
\frac{\widetilde N^{\rm ret}_{0000}(\mathbf k,\mathbf k')}{c^4}
=m_s^2\left[\widetilde p(\mathbf k-\mathbf k')
-\widetilde p(\mathbf k)\widetilde p^*(\mathbf k')\right].
\label{eq:Nk_CQG}
\end{equation}
Using $\widetilde{|\mathbf r|^{-1}}=4\pi/k^2$, the projected covariance is
\begin{align}
\langle \widetilde\xi_\Phi(\mathbf k)\widetilde\xi_\Phi^*(\mathbf k')\rangle_{\rm s}
&=\frac{16\pi^2G^2m_s^2}{k^2k'^2}
\left[\widetilde p(\mathbf k-\mathbf k')
-\widetilde p(\mathbf k)\widetilde p^*(\mathbf k')\right].
\label{eq:Ck_CQG}
\end{align}
The wave-packet form factor couples different momentum modes. The unsmeared DP potential-noise kernel is instead diagonal and proportional to $\delta(\mathbf k-\mathbf k')/k^2$. Table~\ref{tab:comparison} collects the main differences.

\begin{table}[t]
\caption{Einstein--Langevin source covariances for the specified localised one-particle preparations, compared with the phenomenological Di\'osi--Penrose potential-noise model.}
\label{tab:comparison}
\centering
\small
\setlength{\tabcolsep}{8pt}
\begin{tabular}{p{0.18\linewidth}p{0.34\linewidth}p{0.34\linewidth}}
\toprule
\raggedright Feature & \raggedright Specified source preparations & \raggedright DP potential-noise model \tabularnewline
\midrule
\raggedright Origin & \raggedright Connected stress--energy source covariance & \raggedright Postulated effective potential noise \tabularnewline
\raggedright Spatial structure (one particle) & \raggedright $G^2m_s^2[\mathcal I-\phi\phi]$ & \raggedright Smeared Newton kernel (unsmeared form: $G\hbar/r$) \tabularnewline
\raggedright State dependence & \raggedright Explicitly state dependent & \raggedright Postulated state independent and translation invariant \tabularnewline
\raggedright Momentum structure & \raggedright Non-diagonal mode coupling & \raggedright Diagonal, with a spatial-smearing form factor \tabularnewline
\raggedright Fixed total mass & \raggedright Enforced by the connected subtraction & \raggedright Not formulated as a source covariance \tabularnewline
\raggedright Temporal structure & \raggedright Determined by the unequal-time matter spectrum & \raggedright White-noise structure imposed \tabularnewline
\bottomrule
\end{tabular}
\parbox{0.92\textwidth}{\footnotesize The spatial entries compare kernel shapes. The unsmeared DP expression $G\hbar/r$ is the coefficient multiplying $\delta(t-t')$, whereas Eq.~\eqref{eq:Cphi_final_CQG} is an equal-time covariance. Their temporal dimensions are therefore not being identified.}
\end{table}

Figure~\ref{fig:mechanism_cqg} gives a schematic summary of the equal-time and unequal-time results.

\begin{figure}[htbp]
\centering
\includegraphics[width=0.97\textwidth]{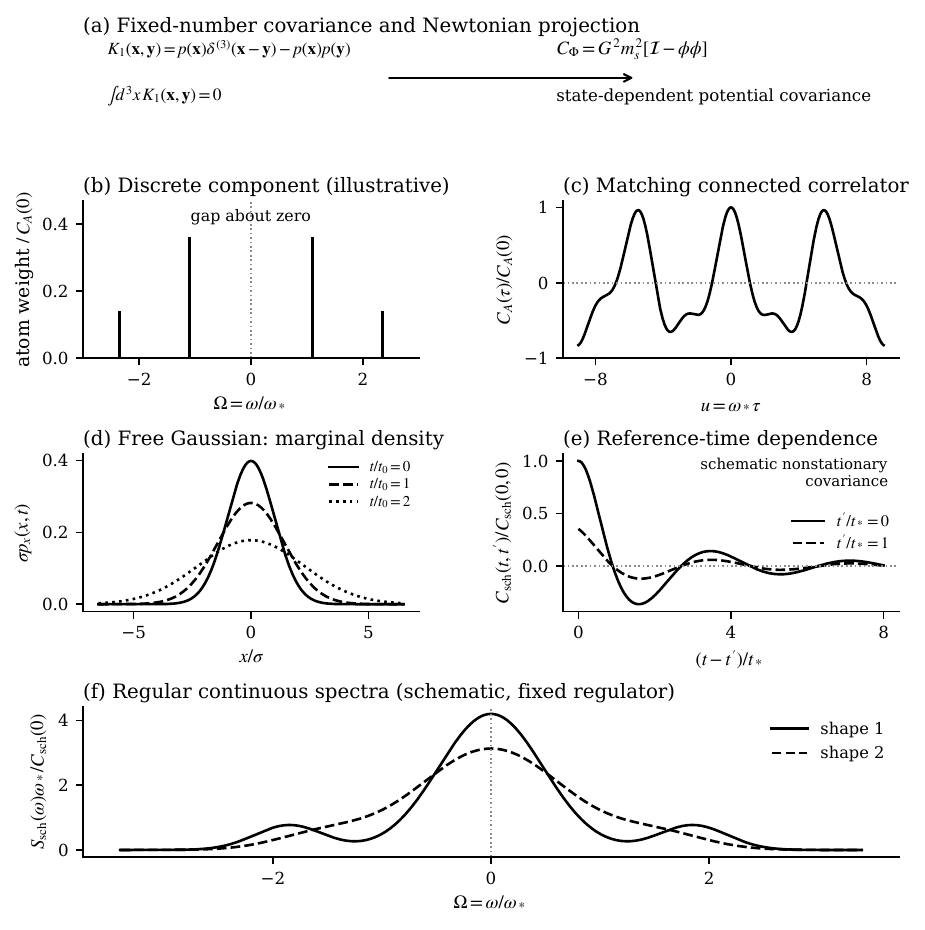}
\begingroup
\linespread{1}\selectfont
\caption{Schematic spatial and temporal content of the source calculation. (a) The fixed-number density covariance and its Newtonian projection. (b,c) A matching discrete spectral component and even correlator; sticks denote atom weights. An isolated eigenstate can also have transitions to a continuum. (d) The Gaussian one-dimensional marginal at $t/t_0=0,1,2$, where $t_0=2m_s\sigma^2/\hbar$. (e) The positive kernel $C_{\rm sch}(t,t')/C_{\rm sch}(0,0)=a(v)a(v')\cos[1.75(v-v')]$, with $v=t/t_*$, $a(v)=(1+v^2)^{-3/4}$, arbitrary $t_*>0$ and $t'/t_*=0,1$. (f) Positive even spectral shapes with equal integrated variance and different values at zero. Panels (e,f) are illustrative, not calculated source kernels or finite-volume spectral measures; the physical continuum prescription is given in Section~\ref{sec:nomarkov_cqg} and Appendix~\ref{app:continuum}.}
\label{fig:mechanism_cqg}
\endgroup
\end{figure}

\section{Source spectra and a static-probe dephasing diagnostic}
\label{sec:nomarkov_cqg}

The equal-time covariance does not determine a dephasing rate.
We specify a static-probe construction and examine bound, freely
spreading and continuum source preparations.

\paragraph{Prescribed source and probe.}
Let the source have particle mass $m_s$, reference state $\sigma_s$, and
Hamiltonian $H_s$. It evolves autonomously, without feedback from the
probe. The probe has mass $m_p$ and two fixed branch projectors
$P_j=|j\rangle\langle j|$ at $\mathbf r_{1,2}=\pm\mathbf R/2$.
We prescribe a real, centred classical Gaussian potential
$\xi_\Phi(\mathbf r,t)$ whose covariance is the symmetrised connected
source-potential covariance. The reference preparation and evolution fix
this law independently of the probe input and subsequent local
operations, and the initial probe and noise are statistically independent.
The probe Hamiltonian is
\begin{equation}
 H_{p,\xi}(t)=\sum_{j=1}^2P_j
 \left\{\epsilon_j(t)+m_p\left[
 \bar\Phi_s(\mathbf r_j,t)+\xi_\Phi(\mathbf r_j,t)\right]\right\},
 \label{eq:prescribed_probe}
\end{equation}
where $\epsilon_j$ is a bare branch energy and $\bar\Phi_s$ is the mean
source potential. The fixed projectors exclude probe transitions.
Tunnelling within $H_s$ instead affects the source covariance.
This construction tests a given source kernel. Matching its symmetrised
two-point function does not determine the quantum response or higher
cumulants required for a microscopic source--probe reduction.

Write $g=Gm_sm_p$ and define the source branch-energy difference by
\begin{equation}
 \begin{split}
 \hat A_s(t)&=-g\int d^3x\,a_{\mathbf R}(\mathbf x)
                   \hat n_s(\mathbf x,t),\\
 a_{\mathbf R}(\mathbf x)&=
 \frac{1}{|\mathbf R/2-\mathbf x|}
 -\frac{1}{|-\mathbf R/2-\mathbf x|},\qquad
 \delta\hat A_s(t)=\hat A_s(t)
 -\operatorname{Tr}[\sigma_s\hat A_s(t)]I .
 \end{split}
 \label{eq:A_CQG}
\end{equation}
The stochastic energy difference
$A_\xi(t)=m_p[\xi_\Phi(\mathbf r_1,t)-\xi_\Phi(\mathbf r_2,t)]$
then has covariance
\begin{equation}
 \mathbb E[A_\xi(t)A_\xi(t')]=C_A(t,t')
 =\tfrac12\operatorname{Tr}\!\left[
 \sigma_s\{\delta\hat A_s(t),\delta\hat A_s(t')\}\right].
 \label{eq:CAdef_CQG}
\end{equation}
For finite second moments this covariance is positive, since its
quadratic form is $\operatorname{Tr}(\sigma_sX^2)\geq0$ for
$X=\sum_i c_i\delta\hat A_s(t_i)$ with real $c_i$.
We assume mean-square time integration, ensured on each finite interval
by $\int_0^T\sqrt{C_A(t,t)}\,dt<\infty$.
Source and probe remain distinct when $m_s=m_p=m$.

For each realisation the probe acquires only a relative phase.
Writing it as $\Theta(T)=\bar\Theta(T)+\eta(T)$, define
\begin{align}
\bar\Theta(T)&=\frac1\hbar\int_0^Tdt\,
[\epsilon_1(t)-\epsilon_2(t)+m_p(\bar\Phi_s(\mathbf r_1,t)-\bar\Phi_s(\mathbf r_2,t))],\nonumber\\
\eta(T)&=\frac1\hbar\int_0^T A_\xi(t)\,dt,\qquad \mathbb E\eta=0.
\label{eq:centered_probe_phase}
\end{align}
Gaussian averaging gives
\begin{equation}
 \begin{split}
 \rho_{12}(T)&=\rho_{12}(0)e^{-i\bar\Theta(T)}e^{-\Gamma(T)},\\
 \Gamma(T)&=\tfrac12\operatorname{Var}\eta(T)
 =\frac{1}{2\hbar^2}\int_0^Tdt\int_0^Tdt'\,C_A(t,t'),\\
 \mathbb E\Theta(T)^2&=\bar\Theta(T)^2+2\Gamma(T).
 \end{split}
 \label{eq:Gamma_finite}
\end{equation}
The coherence magnitude depends on the centred phase variance, rather than the total-phase
second moment used in Eq.~(35) of Ref.~\cite{Bera2015}.
A deterministic phase drift or a finite-window phase criterion alone
does not establish a long-time dephasing rate.

\paragraph{Stationary spectra.}
For a stationary source, write $C_A(t,t')=C_A(t-t')$. Its finite, positive,
even spectral measure $\mu_A$ and its two-sided density, wherever the
measure is absolutely continuous, are defined by
\begin{equation}
 C_A(\tau)=\int_{\mathbb R}e^{-i\omega\tau}\mu_A(d\omega),
 \qquad \mu_A(d\omega)=\frac{S_A(\omega)}{2\pi}\,d\omega.
 \label{eq:source_spectral_measure}
\end{equation}
Here $\mu_A(\mathbb R)=C_A(0)$; $C_A$ has units of energy squared and
$S_A$ has units of energy squared times time. The observation window gives
\begin{equation}
 \begin{split}
 \Gamma(T)&=\frac{1}{\hbar^2}\int\mu_A(d\omega)
             \frac{1-\cos(\omega T)}{\omega^2},\\
 \frac{\Gamma(T)}{T}&=\frac{1}{2\hbar^2}
             \int F_T(\omega)\mu_A(d\omega),\qquad
 F_T(\omega)=T\operatorname{sinc}^2(\omega T/2),
 \end{split}
 \label{eq:spectral_filter}
\end{equation}
where $\operatorname{sinc}x=\sin x/x$ and the first integrand at zero is
$T^2/2$. Discrete spectral atoms must be retained in this formula.

\begin{proposition}
Suppose that the source is in an isolated, non-degenerate eigenstate of
$H_s$, separated by $\hbar\omega_g>0$ from the rest of its spectrum,
and that the branch-difference operator has finite variance in that
state. Then its connected symmetrised spectral measure has no support
in $(-\omega_g,\omega_g)$.
\end{proposition}

\begin{proof}
For $H_s|a\rangle=E_a|a\rangle$, put
$|v\rangle=\delta\hat A_s|a\rangle$ and let $P_K$ be the spectral
resolution of $K=(H_s-E_a)/\hbar$. The spectral theorem gives
$\nu(B)=\langle v|P_K(B)|v\rangle$ and
$\mu_A(B)=[\nu(B)+\nu(-B)]/2$ for every Borel set $B$.
Centring removes the projection onto $|a\rangle$. Non-degeneracy and
the spectral gap exclude all remaining weight from the stated interval.
This argument includes continuous spectral contributions above a
threshold and does not assume that the full spectrum is discrete.
\end{proof}

Within the prescribed Gaussian probe model, the proposition has the
immediate consequence
\begin{equation}
 0\leq\Gamma(T)\leq\frac{2C_A(0)}{\hbar^2\omega_g^2},
 \qquad \lim_{T\to\infty}\frac{\Gamma(T)}T=0.
 \label{eq:zero_CQG}
\end{equation}
Finite-time attenuation is allowed. The bound requires neither decay of $C_A(\tau)$ nor
monotonicity of the coherence magnitude. For comparison, suppose the
\emph{entire} measure in a neighbourhood of zero has a bounded density
continuous at zero, with finite measure outside that neighbourhood.
Since $F_T/(2\pi)$ is an approximate identity, Eq.~\eqref{eq:spectral_filter}
then implies
\begin{equation}
 \lim_{T\to\infty}\frac{\Gamma(T)}T
 =\frac{S_A(0)}{2\hbar^2}.
 \label{eq:regular_spectral_slope}
\end{equation}
This is an asymptotic pure-dephasing slope. An exact time-homogeneous
semigroup would require $\Gamma(T)=\gamma T$ at every time, together
with a time-homogeneous deterministic phase. A zero-frequency atom
$w_0\delta_0$ instead contributes $w_0T^2/(2\hbar^2)$.
The zero-frequency criterion belongs to the static diagonal coupling
in Eq.~\eqref{eq:prescribed_probe}; transition channels can sample
nonzero Bohr frequencies.

For a stationary source doublet,
$\sigma_s=p_+|+\rangle\langle+|+p_-|-\rangle\langle-|$, let
$a_\pm=\langle\pm|\hat A_s|\pm\rangle$,
$w_{\rm t}=|\langle+|\hat A_s|-\rangle|^2$ and
$\omega_{\rm t}=(E_+-E_-)/\hbar\ne0$. Within this two-level truncation,
\begin{equation}
 \begin{split}
 C_A(\tau)&=w_0+w_{\rm t}\cos(\omega_{\rm t}\tau),
 \qquad w_0=p_+p_-(a_+-a_-)^2,\\
 \Gamma(T)&=\frac{w_0T^2}{2\hbar^2}
 +\frac{w_{\rm t}}{\hbar^2\omega_{\rm t}^2}
       [1-\cos(\omega_{\rm t}T)].
 \end{split}
 \label{eq:tunnel_transient}
\end{equation}
The nonzero-frequency term is quadratic at short times and revives after
$2\pi/|\omega_{\rm t}|$. The static term need not revive.
Coherent superpositions of non-degenerate energy eigenstates generally
require $C_A(t,t')$. At exact degeneracy, a state supported in the
eigenspace $P_E$ has the preparation-dependent static weight
$\operatorname{Var}_{\sigma_s}(P_E\hat A_sP_E)$.

\paragraph{A localised free source.}
The unequal-time covariance cannot be reconstructed from the equal-time
density alone. For $H_s=\hat{\mathbf p}^2/(2m_s)$, the propagator for
$t>t'$ is
\begin{equation}
 G_0(\mathbf x,t;\mathbf y,t')
 =\left[\frac{m_s}{2\pi i\hbar(t-t')}\right]^{3/2}
 \exp\!\left[\frac{im_s|\mathbf x-\mathbf y|^2}
                  {2\hbar(t-t')}\right].
 \label{eq:freeprop_CQG}
\end{equation}
For a normalised one-particle wave function $\Psi$ and $p=|\Psi|^2$,
\begin{equation}
 \begin{split}
 \langle\hat n_s(\mathbf x,t)\hat n_s(\mathbf y,t')\rangle
 &=G_0(\mathbf x,t;\mathbf y,t')\Psi^*(\mathbf x,t)\Psi(\mathbf y,t'),\\
 C_n(\mathbf x,t;\mathbf y,t')
 &=\operatorname{Re}[G_0\Psi^*(\mathbf x,t)\Psi(\mathbf y,t')]
       -p(\mathbf x,t)p(\mathbf y,t').
 \end{split}
 \label{eq:free_cn_CQG}
\end{equation}
The equal-time limit is Eq.~\eqref{eq:K1_CQG}, in the distributional
sense. Integrating $C_n$ against
$g^2a_{\mathbf R}(\mathbf x)a_{\mathbf R}(\mathbf y)$ gives $C_A$.

Take an initially unchirped isotropic Gaussian centred at the branch
midpoint, with zero mean momentum and variance $\sigma^2$ per Cartesian
direction. Its density width is
$\sigma_t=\sigma[1+(t/t_0)^2]^{1/2}$, where
$t_0=2m_s\sigma^2/\hbar$. Parity gives $\langle\hat A_s(t)\rangle=0$.
For $R=|\mathbf R|>0$, the exact spatial identity
$\int d^3x\,a_{\mathbf R}(\mathbf x)^2=4\pi R$ yields
\begin{equation}
 \begin{split}
 V_A(t)&\equiv\operatorname{Var}_t(\hat A_s)
 \leq4\pi g^2R(2\pi\sigma_t^2)^{-3/2},\\
 V_A(t)&=4\pi g^2R(2\pi\sigma_t^2)^{-3/2}[1+o(1)]
 \quad(t\to\infty).
 \end{split}
 \label{eq:free_variance_asymptotic}
\end{equation}
The bound holds at all times; dominated convergence gives the
asymptotic equality. Since $|C_A(t,t')|\leq\sqrt{V_A(t)V_A(t')}$ and
$\int_0^\infty\sqrt{V_A(t)}\,dt<\infty$, the covariance is absolutely
integrable over the positive time quadrant. Consequently,
\begin{equation}
 \Gamma(T)\longrightarrow\Gamma_\infty\in[0,\infty),
 \qquad \Gamma(T)/T\longrightarrow0.
 \label{eq:free_finite_limit}
\end{equation}
The exponent need not approach its limit monotonically. Nonstationarity
or saturation alone does not establish completely positive
indivisibility or information backflow.

\paragraph{Stationary continuum preparations.}
A continuum limit requires a separate normalisation and limiting
prescription. In a periodic cube of volume $V=L^3$, take
$\langle\mathbf x|\mathbf k\rangle=V^{-1/2}e^{i\mathbf k\cdot\mathbf x}$,
$\mathbf k\in(2\pi/L)\mathbb Z^3$, and a momentum-diagonal state
$\sigma_{s,V}=\sum_{\mathbf k}P_{\mathbf k,V}
|\mathbf k\rangle\langle\mathbf k|$, with
$\sum_{\mathbf k}P_{\mathbf k,V}=1$. This specifies the preparation
also within degenerate energy eigenspaces. For a zero-mean periodic
operator with $A_{\mathbf k'\mathbf k}=
\widetilde A(\mathbf k'-\mathbf k)/V$ and $\widetilde A(0)=0$,
\begin{equation}
 \mu_{A,V}(d\omega)=\frac{1}{2V^2}
 \sum_{\mathbf k,\mathbf k'}P_{\mathbf k,V}
 |\widetilde A(\mathbf k'-\mathbf k)|^2
 [\delta_{\Delta\omega_{\mathbf k'\mathbf k}}(d\omega)
 +\delta_{-\Delta\omega_{\mathbf k'\mathbf k}}(d\omega)],
 \label{eq:box_spectral_measure}
\end{equation}
where $\Delta\omega_{\mathbf k'\mathbf k}=
(E_{\mathbf k'}-E_{\mathbf k})/\hbar$ and
$E_{\mathbf k}=\hbar^2k^2/(2m_s)$. Distinct momenta of equal energy
can produce a zero-frequency atom, rather than a regular spectral density.
Equation~\eqref{eq:spectral_filter} resolves frequencies on a scale
$1/T$. We take the large-box limit at fixed $T$ and fixed infrared
prescription before extracting a long-time slope.

For a smooth normalised momentum density $f$ with controlled tails,
$P_{\mathbf k,V}\simeq(2\pi)^3f(\mathbf k)/V$. The resulting leading
continuum coefficient is
\begin{equation}
 S_{A,V}^{\rm cont}(0)=\frac{2\pi\hbar}{V(2\pi)^3}
 \int d^3k\,f(\mathbf k)\int d^3k'\,
 |\widetilde A(\mathbf k'-\mathbf k)|^2
 \delta(E_{\mathbf k'}-E_{\mathbf k}).
 \label{eq:continuum_CQG}
\end{equation}
This is not the density of the exact finite-box measure.
Its $1/V$ factor vanishes for a single particle in the dilution limit.
For independent additive sources, the different limit $N/V\to n$
replaces $1/V$ by the number density $n$; general quantum statistics or
interparticle correlations require their own density correlator.
A specified incident flux likewise supplies additional preparation data.

For a directional preparation $\mathbf k=k\hat{\mathbf z}$, $k>0$,
choose the regulated Newtonian operator
\begin{equation}
 \widetilde A_{\kappa,\ell}(\mathbf q)
 =8\pi i g\,
 \frac{\sin(\mathbf q\cdot\mathbf R/2)}{q^2+\kappa^2}
 e^{-\ell^2q^2/2},\qquad \kappa>0.
 \label{eq:continuum_regulator}
\end{equation}
Here $\kappa$ specifies an infrared prescription and $\ell\geq0$
spatial smoothing; no physical screening mechanism is assumed.
Writing $R_\perp$ for the separation transverse to $\mathbf k$ and
$v=\hbar k/m_s$, the elastic shell gives
\begin{equation}
 S_{A,V}^{\rm cont}(0\mid\mathbf k)
 =\frac{4\pi g^2R_\perp^2}{Vv}\ln\frac{2k}{\kappa}+O(1)
 \qquad(\kappa\downarrow0).
 \label{eq:continuum_IR_log}
\end{equation}
Parallel separation removes this logarithm, although the unscreened
and zero-frequency limits need not commute.
Appendix~\ref{app:continuum} derives the finite-box limit, the angular
dependence, and the additional forward contribution in the parallel case.
These number-conserving, non-relativistic examples therefore depend on
preparation, normalisation, geometry and the limiting prescription.

\section{No-signalling for fixed stochastic channels}\label{sec:nosig_cqg}

The origin of a noise kernel and the no-signalling property of a prescribed stochastic channel are logically separate. The following statements assume that the classical noise process is initially independent of the quantum systems and fixed independently of local choices, and that local operations are non-selective and trace preserving. For white noise, let $D(\mathbf r,\mathbf r')$ be the coefficient of $\delta(t-t')$ in the potential covariance, and diagonalise this positive kernel as $D(\mathbf r,\mathbf r')=\sum_\alpha\lambda_\alpha f_\alpha(\mathbf r)f_\alpha(\mathbf r')$. For systems in disjoint regions $\mathcal R_A$ and $\mathcal R_B$, the coupling is additive,
\begin{equation}
\hat L_\alpha^{\rm tot}=\hat L_\alpha^{(A)}\otimes I_B+I_A\otimes\hat L_\alpha^{(B)}.
\end{equation}
With a deterministic Hamiltonian that is also additive and local to the two subsystems, the Markovian dephasing dissipator is~\cite{Breuer2002,Lindblad1976,Gorini1976}
\begin{equation}
\mathcal D_{AB}\rho_{AB}=-\frac{1}{2\hbar^2}\sum_\alpha\lambda_\alpha
[\hat L_\alpha^{\rm tot},[\hat L_\alpha^{\rm tot},\rho_{AB}]].
\label{eq:DAB_CQG}
\end{equation}
A cross term obeys
\begin{align}
\operatorname{Tr}_A[\hat L_A\otimes I,[I\otimes\hat L_B,\rho_{AB}]]
&=[\hat L_B,\operatorname{Tr}_A(\hat L_A\rho_{AB}-\rho_{AB}\hat L_A)]=0,
\end{align}
because the partial trace is cyclic for Alice-local operators. Hence
\begin{equation}
\operatorname{Tr}_A\mathcal D_{AB}\rho_{AB}=\mathcal D_B\rho_B.
\label{eq:nosig_CQG}
\end{equation}
Bob's unconditional reduced dynamics is unchanged by a local completely positive trace-preserving map $\mathcal E_A$ on Alice's subsystem. The same identity follows from a Kraus representation of $\mathcal E_A$.

The same argument applies to the prescribed coloured noise of Section~\ref{sec:nomarkov_cqg}. If each realisation has the additive Hamiltonian
\begin{equation}
H_\xi(t)=H_{A,\xi}(t)\otimes I_B+I_A\otimes H_{B,\xi}(t),
\end{equation}
then $U_\xi=U_{A,\xi}\otimes U_{B,\xi}$ and
\begin{equation}
\rho_B(t)=\mathbb E_\xi\left[U_{B,\xi}(t)\rho_B(0)U_{B,\xi}^\dagger(t)\right].
\end{equation}
No Markov approximation is required. These statements concern prescribed linear channels. They do not cover deterministic Schr\"odinger--Newton nonlinearity or a self-consistently updated nonlocal kernel $C[\rho]$, which require a linear ensemble-level completion, a measurement-feedback construction, or covariantly retarded dynamics~\cite{Kafri2014,Tilloy2016,Tilloy2017,Oppenheim2023}.

\section{Fixed-number many-body covariance}\label{sec:manybody_cqg}

For a fixed-$N$ bosonic state, define
\begin{align}
n_1(\mathbf x)&=\langle \hat n(\mathbf x)\rangle,\nonumber\\
n_2(\mathbf x,\mathbf y)&=\langle\hat\psi^\dagger(\mathbf x)\hat\psi^\dagger(\mathbf y)
\hat\psi(\mathbf x)\hat\psi(\mathbf y)\rangle.
\end{align}
The exact connected covariance is
\begin{equation}
K_N(\mathbf x,\mathbf y)
=n_1(\mathbf x)\delta^{(3)}(\mathbf x-\mathbf y)
+n_2(\mathbf x,\mathbf y)-n_1(\mathbf x)n_1(\mathbf y).
\label{eq:KN_CQG}
\end{equation}
Since $\int n_1=N$ and $\int d^3y\,n_2(\mathbf x,\mathbf y)=(N-1)n_1(\mathbf x)$, every fixed-$N$ state satisfies
\begin{equation}
\int d^3y\,K_N(\mathbf x,\mathbf y)=0,\qquad\int d^3x\,d^3y\,K_N(\mathbf x,\mathbf y)=0.
\end{equation}
For the Hartree state $|N:\psi\rangle=(\hat a_\psi^\dagger)^N|0\rangle/\sqrt{N!}$,
\begin{equation}
K_N^{\rm H}(\mathbf x,\mathbf y)
=N\left[p(\mathbf x)\delta^{(3)}(\mathbf x-\mathbf y)-p(\mathbf x)p(\mathbf y)\right].
\label{eq:KH_CQG}
\end{equation}
The Hartree source covariance grows linearly with $N$. Collective branch-occupation correlations can change this scaling. Let $u_L,u_R$ be normalised orbitals with disjoint spatial supports, $p_\alpha=|u_\alpha|^2$, and $|\alpha_N\rangle=(a_\alpha^\dagger)^N|0\rangle/\sqrt{N!}$. Compare the number cat and the equal-weight mixture
\begin{equation}
|C_\varphi\rangle=\frac{|L_N\rangle+e^{i\varphi}|R_N\rangle}{\sqrt2},\qquad
\rho_M=\frac12\left(|L_N\rangle\langle L_N|+|R_N\rangle\langle R_N|\right).
\label{eq:cat_and_mixture}
\end{equation}
For every $N\ge1$, their full equal-time covariances coincide:
\begin{align}
K_C(\mathbf x,\mathbf y)=K_M(\mathbf x,\mathbf y)
={}&\frac N2[p_L(\mathbf x)+p_R(\mathbf x)]\delta^{(3)}(\mathbf x-\mathbf y)\nonumber\\
&-\frac N2[p_L(\mathbf x)p_L(\mathbf y)+p_R(\mathbf x)p_R(\mathbf y)]\nonumber\\
&+\frac{N^2}{4}\Delta p(\mathbf x)\Delta p(\mathbf y),\qquad\Delta p=p_L-p_R.
\label{eq:NOON_covariance}
\end{align}
The contact term requires the full field commutator, including unoccupied modes. The mass-density covariance is $m_s^2K$, and both coordinate integrals vanish. The $N^2$ term is already present in the mixture and reflects collective branch occupation. It is not a witness of branch coherence. This equality assumes disjoint supports and concerns equal times. Subsequent source dynamics can make coherence observable in unequal-time correlations.

A total-mass description can capture the leading equal-time scaling for this collective occupation coordinate, but it does not determine its temporal spectrum, internal form factor or spatial coupling. Replacing $m_s$ by $Nm_s$ in a single-particle dephasing rate is therefore not justified by Eq.~\eqref{eq:NOON_covariance}.

\section{Discussion}\label{sec:discussion_cqg}

The results separate source spectral properties, comparison of potential kernels and construction of a matter channel. The spectral-gap result gives a bounded probe exponent only for the specified independent static probe and prescribed Gaussian noise law. The localised one-particle potential covariance depends on the source preparation and differs from the spatial and temporal law postulated in the DP model. Neither statement is a microscopic derivation of the original matter system's self-decoherence.

A stationary doublet may carry a zero-frequency atom and a tunnel line, with different long-time behaviour. The specified freely spreading Gaussian has an absolutely integrable two-time branch covariance and a finite limiting probe exponent. Continuum results depend on density, momentum distribution, branch orientation, infrared prescription and limiting order. This spectral classification builds on the connected and nonwhite correlations obtained in Ref.~\cite{Bera2015}.

Particle-number conservation alone does not prohibit Markovian dynamics. For bounded Hermitian $F$ with $[F,\hat N]=[H,\hat N]=0$, the generator
\begin{equation}
\mathcal L\rho=-\frac{i}{\hbar}[H,\rho]-\frac{\kappa_F}{2}[F,[F,\rho]],\qquad\kappa_F\ge0,
\label{eq:number_conserving_generator}
\end{equation}
preserves every fixed-number sector. The spectral-gap statement therefore cannot be extended to all number-conserving non-relativistic preparations. The long-time slope in the static diagnostic must be distinguished from a finite-window fit and an exact time-homogeneous semigroup.

A full EL-to-DP matching requires four further steps. The source and response must yield the required temporal spectrum, including its low-frequency behaviour. Their spatial projection must have the stipulated form, translation properties, spatial smearing and degree of preparation independence. The noise law must be specified as fixed, externally prepared or dynamically updated. The coupling and approximation must also justify the reduced matter channel, including response effects not determined by a symmetrised covariance alone. Antiparticles, relativistic pair and vacuum contributions, subleading stress--energy corrections, radiative metric modes and retarded backreaction lie outside the present leading-source calculation. Thermal or other continuum preparations within non-relativistic physics can also alter the spectrum.

Bounds on DP-type noise constrain that specified phenomenological channel. An equal-time EL covariance alone does not predict the same observable. A composite-body prediction additionally needs the many-body state, internal form factor, centre-of-mass width and temporal correlations. The conditional cutoff and width estimate in Appendix~\ref{app:uv_markov} is not a rigid-body formula. The fixed-channel no-signalling proof likewise applies to operation-independent noise and local product evolution at the unconditional ensemble level. It does not establish relativistic causality for an instantaneous interaction or cover a self-consistently updated nonlocal kernel.

\section{Conclusion}\label{sec:conclusion_cqg}

For the specified closed non-relativistic sources, the fixed-number density covariance contains a contact term and a connected subtraction that keeps the total mass constant. Its Newtonian projection and temporal spectrum depend on the source preparation. Equal-time collective occupation correlations can scale as $N^2$ without distinguishing a coherent number cat from its incoherent branch mixture.

An isolated non-degenerate source eigenstate with finite variance and a positive spectral gap gives a bounded dephasing exponent in the separately defined static Gaussian probe diagnostic. The specified freely spreading Gaussian also has a finite limiting exponent, while doublets and continuum preparations require their own temporal and limiting analysis. These results show why nonzero equal-time stress--energy variance does not by itself establish a universal Di\'osi--Penrose matter channel. The vanishing long-time slope applies under the stated source and probe assumptions, rather than throughout number-conserving non-relativistic physics.

\appendix

\section{Two-peak Gaussian covariance}\label{app:gaussian}

This appendix concerns a one-particle two-peak density, distinct from the many-body number cat in Section~\ref{sec:manybody_cqg}. Let
\begin{equation}
p_0(\mathbf u)=(2\pi\sigma^2)^{-3/2}\exp[-u^2/(2\sigma^2)]
\end{equation}
be a normalised three-dimensional Gaussian. The radial integrals required by Eq.~\eqref{eq:Cphi_final_CQG} are
\begin{align}
\int d^3u\,\frac{p_0(\mathbf u)}{|\mathbf u|}
&=\sqrt{\frac{2}{\pi}}\frac{1}{\sigma},\nonumber\\
\int d^3u\,\frac{p_0(\mathbf u)}{|\mathbf u|^2}
&=\frac{1}{\sigma^2}.
\label{eq:Gaussian_integrals}
\end{align}
For
\begin{equation}
p(\mathbf x)=\frac12\left[p_0(\mathbf x-\mathbf R/2)
+p_0(\mathbf x+\mathbf R/2)\right],
\qquad R\gg\sigma,
\end{equation}
the first terms of the large-separation expansion are
\begin{align}
C_\Phi(\mathbf R/2,\mathbf R/2)
&=G^2m_s^2\left[
\frac{\pi-1}{2\pi\sigma^2}
-\sqrt{\frac{2}{\pi}}\frac{1}{2\sigma R}
+\frac{1}{4R^2}+\cdots\right],\nonumber\\
C_\Phi(\mathbf R/2,-\mathbf R/2)
&=G^2m_s^2\left[
-\frac{1}{2\pi\sigma^2}
+\sqrt{\frac{2}{\pi}}\frac{1}{2\sigma R}
-\frac{1}{4R^2}+\cdots\right].
\label{eq:Gaussian_expansion}
\end{align}
Consequently,
\begin{equation}
\Delta C_h=G^2m_s^2\left[
\frac{1}{2\sigma^2}
-\sqrt{\frac{2}{\pi}}\frac{1}{\sigma R}
+\frac{1}{2R^2}+\cdots\right].
\label{eq:DeltaC_expansion}
\end{equation}
The leading term is independent of $R$ because the self contribution is controlled by the packet width. The displayed corrections quantify the approach to the $R\gg\sigma$ limit used in Appendix~\ref{app:uv_markov}.

\section{An imposed temporal spectrum and its long-time slope}\label{app:uv_markov}

This appendix imposes a temporal spectrum independently of the source dynamics. The resulting estimate depends on additional temporal and width assumptions, not on the spectral-gap proposition. Consider
\begin{align}
S_\xi(\mathbf r,\mathbf r';\omega)&=\frac{C_\Phi(\mathbf r,\mathbf r')}{\omega_c}F(\omega/\omega_c),\nonumber\\
F(\Omega)&\ge0,\qquad F(-\Omega)=F(\Omega),\qquad
\int\frac{d\Omega}{2\pi}F(\Omega)=1.
\label{eq:Smodel_CQG}
\end{align}
Let $h(\tau)=\int d\Omega\,F(\Omega)e^{-i\omega_c\Omega\tau}/(2\pi)$, so that $h(0)=1$. If $F$ is integrable and continuous at zero, the filtered long-time limit is
\begin{equation}
\tau_{\rm eff}:=\lim_{T\to\infty}\int_0^T\left(1-\frac{\tau}{T}\right)h(\tau)d\tau
=\frac{F(0)}{2\omega_c}.
\label{eq:taueff_CQG}
\end{equation}
It agrees with the ordinary one-sided correlation integral when that integral exists. Choosing $\omega_c=m_sc^2/\hbar$ is an extra cutoff assumption, not a consequence of the retained covariance. The flat-band choice $F(\Omega)=\pi\Theta(1-|\Omega|)$ gives $\tau_{\rm eff}=\pi/(2\omega_c)$. Its finite bandwidth gives a coloured correlator. The following coefficient describes the long-time slope, not exact white-noise dynamics at every time.

For the two-peak one-particle Gaussian density in Appendix~\ref{app:gaussian}, with $R\gg\sigma$, Eqs.~\eqref{eq:Gaussian_expansion} and \eqref{eq:DeltaC_expansion} give
\begin{align}
\Delta C_h&=C_\Phi(\mathbf R/2,\mathbf R/2)-C_\Phi(\mathbf R/2,-\mathbf R/2)
\simeq\frac{G^2m_s^2}{2\sigma^2},\nonumber\\
D_\Phi(0)&=2\Delta C_h,\qquad
\gamma_{\rm cut}=\frac{m_p^2}{\hbar^2}\tau_{\rm eff}D_\Phi(0)
\simeq\frac{\pi G^2m_sm_p^2}{2c^2\hbar\sigma^2}.
\label{eq:gammaUV_CQG}
\end{align}
Here $D_\Phi(0)$ is the equal-time branch-difference potential variance; the probe mass enters when it is converted to an energy variance. If one also inserts the Gaussian Schr\"odinger--Newton variational width~\cite{Moroz1998,Giulini2012},
\begin{equation}
\sigma_0=\frac{3\sqrt\pi\hbar^2}{2Gm_s^3},\qquad
\gamma_{\rm width}\simeq\frac{2G^4m_s^7m_p^2}{9c^2\hbar^5},
\label{eq:m9_CQG}
\end{equation}
then the equal-mass specialisation is proportional to $m^9$. This power follows only after imposing both the cutoff and the width. Neither is derived from the source-spectrum proposition, and the estimate does not apply to ordinary composite bodies by a total-mass substitution.

\section{Continuum normalisation and infrared limits}
\label{app:continuum}

The periodic operator underlying Eq.~\eqref{eq:box_spectral_measure} is
defined by the Fourier convention
\begin{equation}
 A_{V,\kappa,\ell}(\mathbf x)=\frac1V
 \sum_{\mathbf q\in(2\pi/L)\mathbb Z^3}
 \widetilde A_{\kappa,\ell}(\mathbf q)e^{i\mathbf q\cdot\mathbf x},
 \qquad \widetilde A_{\kappa,\ell}(0)=0,
 \label{eq:periodic_operator}
\end{equation}
with Eq.~\eqref{eq:continuum_regulator}. This fixes the periodic
boundary prescription and the overall sign corresponding to
Eq.~\eqref{eq:A_CQG}; the sign drops out of the covariance.
At fixed $\kappa>0$, the large-box limit eventually has $\kappa L\gg1$.
For a smooth momentum density, the finite-time Riemann sum gives
\begin{equation}
 \begin{split}
 \mathcal W_\kappa(T)&:=
 \lim_{V\to\infty}\frac{V\Gamma_{V,\kappa}(T)}T\\
 &=\frac{1}{2\hbar^2(2\pi)^3}
 \int d^3k\,f(\mathbf k)\int d^3q\,
 |\widetilde A_{\kappa,\ell}(\mathbf q)|^2
 F_T\!\left[\frac{\hbar}{2m_s}
                 (q^2+2\mathbf k\cdot\mathbf q)\right].
 \end{split}
 \label{eq:continuum_finite_time}
\end{equation}
The same expression with fixed $\mathbf k$ defines a directional
monoenergetic sequence of box-normalised states, not a normalisable infinite-volume
one-particle state. The factor $V$ extracts the dilute one-particle
coefficient. Independent additive sources with $N/V\to n$ have
$\Gamma(T)/T=n\mathcal W_\kappa(T)$ in that limit.

At fixed $L$, free energy differences are integer multiples of
$\Omega_L=2\pi^2\hbar/(m_sL^2)$. If $w_{0,V}$ is the zero-frequency
atom and the total variance is finite, then
\begin{equation}
 \begin{split}
 \Gamma_V(T)&=\frac{w_{0,V}T^2}{2\hbar^2}
                      +\Gamma_{V,\ne0}(T),\\
 0\leq\Gamma_{V,\ne0}(T)&\leq
 \frac{2\mu_{A,V}(\mathbb R\setminus\{0\})}{\hbar^2\Omega_L^2}.
 \end{split}
 \label{eq:box_long_time}
\end{equation}
The nonzero-frequency contribution is periodic with period
$2\pi/\Omega_L$. A fixed box therefore gives either a quadratic static
term or a vanishing long-time slope. The finite-time filter resolves
frequencies on a scale of order $1/T$; the continuum limit in
Eq.~\eqref{eq:continuum_finite_time} is taken at fixed $T$ before
increasing the observation time.

\paragraph{Elastic shell and direction.}
For incident $\mathbf k=k\hat{\mathbf z}$, write
$\mathbf R=\mathbf R_\perp+R_\parallel\hat{\mathbf z}$,
$q=2k\sin(\theta/2)$, and $Q=2k$. The azimuthal integral of
$\sin^2(\mathbf q\cdot\mathbf R/2)$ equals $\pi D(q)$, where
\begin{equation}
 D(q)=1-\cos\!\left(\frac{R_\parallel q^2}{2k}\right)
 J_0\!\left(R_\perp q\sqrt{1-\frac{q^2}{4k^2}}\right),
 \label{eq:angular_shell_factor}
\end{equation}
and $J_0$ is the Bessel function. The denominator prescription in
Eq.~\eqref{eq:continuum_regulator} gives
\begin{equation}
 S_{A,V}^{\rm cont}(0\mid\mathbf k)=\frac{16\pi g^2}{Vv}\int_0^Q
 \frac{q e^{-\ell^2q^2}D(q)}{(q^2+\kappa^2)^2}\,dq .
 \label{eq:elastic_shell_integrals}
\end{equation}
Since $e^{-\ell^2q^2}D(q)=R_\perp^2q^2/4+O(q^4)$ and
$\int_0^Q q^3(q^2+\kappa^2)^{-2}dq=\ln(Q/\kappa)+O(1)$ as $\kappa\downarrow0$,
the logarithmic coefficient is $C_{\log}=4\pi g^2R_\perp^2/(Vv)$,
as in Eq.~\eqref{eq:continuum_IR_log}.
Spatial smoothing leaves the leading infrared coefficient unchanged.
The elastic shell is bounded in $q$, so its upper endpoint is finite
even at $\ell=0$. For $R_\perp=0$, $D(q)=O(q^4)$ and the
unscreened elastic integral converges. At $\ell=0$, this regulated-first
elastic limit is
\begin{equation}
 S_{\parallel,\rm el}(0)=\frac{4\pi g^2R}{Vvk}
 \left[\operatorname{Si}(2kR)
 -\frac{1-\cos(2kR)}{2kR}\right],
 \label{eq:parallel_elastic}
\end{equation}
where $\operatorname{Si}(z)=\int_0^z\sin u\,du/u$.

\paragraph{Forward contribution away from the elastic shell.}
For parallel separation and sufficiently small nonzero $|\omega|$,
both shells $K_s=(k^2+2m_s s\omega/\hbar)^{1/2}$, $s=\pm1$, exist.
The symmetrised spectrum is
\begin{equation}
 VS_\parallel(\omega;\kappa)=\frac{8\pi g^2}{v}
 \sum_{s=\pm1}\int_{(K_s-k)^2}^{(K_s+k)^2}
 \frac{e^{-\ell^2x}
 \sin^2\!\left[\frac R2\left(\frac{s\omega}{v}
                                  -\frac{x}{2k}\right)\right]}
 {(x+\kappa^2)^2}\,dx .
 \label{eq:parallel_offshell}
\end{equation}
Put $\delta=|\omega|/v$. In the shrinking forward region $x\sim\delta^2$,
the leading integral is proportional to
$\delta^2\int_{\delta^2}^{\eta}(x+\kappa^2)^{-2}dx$.
Taking $\kappa/\delta\to r$ before $\eta\downarrow0$ gives
$(1+r^2)^{-1}$, while the outer region gives the elastic contribution.
Consequently,
\begin{equation}
 VS_\parallel(\omega;\kappa)\longrightarrow
 \mathcal S_{\rm el}+\frac{4\pi g^2R^2}{v(1+r^2)},
 \qquad \mathcal S_{\rm el}:=
 \lim_{\kappa\downarrow0}VS_\parallel(0;\kappa).
 \label{eq:parallel_forward_limit}
\end{equation}
After the continuum limit, the corresponding ordered limits are
\begin{equation}
 \begin{split}
 \lim_{\kappa\downarrow0}\lim_{T\to\infty}\mathcal W_\kappa(T)
 &=\frac{\mathcal S_{\rm el}}{2\hbar^2},\\
 \lim_{T\to\infty}\lim_{\kappa\downarrow0}\mathcal W_\kappa(T)
 &=\frac{\mathcal S_{\rm el}}{2\hbar^2}
    +\frac{2\pi g^2R^2}{v\hbar^2}.
 \end{split}
 \label{eq:parallel_ordered_limits}
\end{equation}
A density may be changed at one point without changing its measure.
The unscreened time filter samples the local continuous extension, not
an isolated assignment at exactly $\omega=0$. The finite-time
unscreened limit exists because the total branch-difference variance
is finite. Nevertheless, regulator removal and long-time extraction
remain distinct even when the elastic integral is finite.

\begin{acknowledgments}
This work is supported by the National Key R\&D Program of China under Grants No. 2025YFE0217300 and No. 2024YFC2206902.
\end{acknowledgments}

\section*{Data availability statement}
The analytical expressions needed to reproduce the results are given in the article. Figure-generation code, numerical verification scripts and reference outputs are supplied as supplementary material.

\section*{Conflict of interest}
The authors declare no competing interests.


\begin{thebibliography}{50}

\bibitem{Moller1962} C. M{\o}ller, in \textit{Les Th\'eories Relativistes de la Gravitation} (CNRS, Paris, 1962).
\bibitem{Rosenfeld1963} L. Rosenfeld, Nucl. Phys. \textbf{40}, 353 (1963). \url{https://doi.org/10.1016/0029-5582(63)90279-7}
\bibitem{Kiefer1994} C. Kiefer, ``The semiclassical approximation to quantum gravity,'' in \textit{Canonical Gravity: From Classical to Quantum}, Lecture Notes in Physics Vol. 434, edited by J. Ehlers and H. Friedrich (Springer, Berlin, Heidelberg, 1994), pp. 170--212. \url{https://doi.org/10.1007/3-540-58339-4_19}
\bibitem{Diosi1984} L. Di\'osi, Phys. Lett. A \textbf{105}, 199 (1984). \url{https://doi.org/10.1016/0375-9601(84)90397-9}
\bibitem{Penrose1996} R. Penrose, Gen. Relativ. Gravit. \textbf{28}, 581 (1996). \url{https://doi.org/10.1007/BF02105068}
\bibitem{Giulini2012} D. Giulini and A. Gro{\ss}ardt, Class. Quantum Grav. \textbf{29}, 215010 (2012). \url{https://doi.org/10.1088/0264-9381/29/21/215010}
\bibitem{Moroz1998} I. M. Moroz, R. Penrose, and P. Tod, Class. Quantum Grav. \textbf{15}, 2733 (1998). \url{https://doi.org/10.1088/0264-9381/15/9/019}
\bibitem{Giulini2014} D. Giulini and A. Gro{\ss}ardt, New J. Phys. \textbf{16}, 075005 (2014). \url{https://doi.org/10.1088/1367-2630/16/7/075005}
\bibitem{Anastopoulos2014} C. Anastopoulos and B. L. Hu, New J. Phys. \textbf{16}, 085007 (2014). \url{https://doi.org/10.1088/1367-2630/16/8/085007}
\bibitem{Hu2000} W. Hu, R. Barkana, and A. Gruzinov, Phys. Rev. Lett. \textbf{85}, 1158 (2000). \url{https://doi.org/10.1103/PhysRevLett.85.1158}
\bibitem{Guzman2004} F. S. Guzm\'an and L. A. Ure\~na-L\'opez, Phys. Rev. D \textbf{69}, 124033 (2004). \url{https://doi.org/10.1103/PhysRevD.69.124033}
\bibitem{Kopp2017} M. Kopp, K. Vattis, and C. Skordis, Phys. Rev. D \textbf{96}, 123532 (2017). \url{https://doi.org/10.1103/PhysRevD.96.123532}
\bibitem{Bekenstein2015} R. Bekenstein et al., Nat. Phys. \textbf{11}, 872 (2015). \url{https://doi.org/10.1038/nphys3451}
\bibitem{Roger2016} T. Roger et al., Nat. Commun. \textbf{7}, 13492 (2016). \url{https://doi.org/10.1038/ncomms13492}
\bibitem{Paz2025} O. Paz et al., Nat. Commun. \textbf{16}, 4113 (2025). \url{https://doi.org/10.1038/s41467-025-59384-z}
\bibitem{Bahrami2014} M. Bahrami, A. Gro{\ss}ardt, S. Donadi, and A. Bassi, New J. Phys. \textbf{16}, 115007 (2014). \url{https://doi.org/10.1088/1367-2630/16/11/115007}
\bibitem{Bielinska2025} M. E. Bieli\'nska, M. Eckstein, and P. Horodecki, New J. Phys. \textbf{27}, 053002 (2025). \url{https://doi.org/10.1088/1367-2630/adcae1}
\bibitem{Diosi2007} L. Di\'osi, J. Phys. A \textbf{40}, 2989 (2007). \url{https://doi.org/10.1088/1751-8113/40/12/S07}
\bibitem{Bassi2003} A. Bassi and G. C. Ghirardi, Phys. Rep. \textbf{379}, 257 (2003). \url{https://doi.org/10.1016/S0370-1573(03)00103-0}
\bibitem{Bassi2013} A. Bassi et al., Rev. Mod. Phys. \textbf{85}, 471 (2013). \url{https://doi.org/10.1103/RevModPhys.85.471}
\bibitem{Bassi2017} A. Bassi, A. Gro{\ss}ardt, and H. Ulbricht, Class. Quantum Grav. \textbf{34}, 193002 (2017). \url{https://doi.org/10.1088/1361-6382/aa864f}
\bibitem{Nimmrichter2014} S. Nimmrichter, K. Hornberger, and K. Hammerer, Phys. Rev. Lett. \textbf{113}, 020405 (2014). \url{https://doi.org/10.1103/PhysRevLett.113.020405}
\bibitem{Carlesso2022} M. Carlesso et al., Nat. Phys. \textbf{18}, 243 (2022). \url{https://doi.org/10.1038/s41567-021-01489-5}
\bibitem{Donadi2021} S. Donadi et al., Nat. Phys. \textbf{17}, 74 (2021). \url{https://doi.org/10.1038/s41567-020-1008-4}
\bibitem{Vinante2020} A. Vinante et al., Phys. Rev. Lett. \textbf{125}, 100404 (2020). \url{https://doi.org/10.1103/PhysRevLett.125.100404}
\bibitem{Fadel2023} M. Fadel, arXiv:2305.04780 (2023). \url{https://doi.org/10.48550/arXiv.2305.04780}
\bibitem{Tilloy2019} A. Tilloy, J. Phys. Conf. Ser. \textbf{1275}, 012006 (2019). \url{https://doi.org/10.1088/1742-6596/1275/1/012006}
\bibitem{Asprea2021} L. Asprea, G. Gasbarri, and A. Bassi, Phys. Rev. D \textbf{103}, 104041 (2021). \url{https://doi.org/10.1103/PhysRevD.103.104041}
\bibitem{Anastopoulos2022} C. Anastopoulos and B. L. Hu, AVS Quantum Sci. \textbf{4}, 015602 (2022). \url{https://doi.org/10.1116/5.0077536}
\bibitem{Kafri2014} D. Kafri, J. M. Taylor, and G. J. Milburn, New J. Phys. \textbf{16}, 065020 (2014). \url{https://doi.org/10.1088/1367-2630/16/6/065020}
\bibitem{Tilloy2016} A. Tilloy and L. Di\'osi, Phys. Rev. D \textbf{93}, 024026 (2016). \url{https://doi.org/10.1103/PhysRevD.93.024026}
\bibitem{Tilloy2017} A. Tilloy and L. Di\'osi, Phys. Rev. D \textbf{96}, 104045 (2017). \url{https://doi.org/10.1103/PhysRevD.96.104045}
\bibitem{Anastopoulos2013} C. Anastopoulos and B. L. Hu, Class. Quantum Grav. \textbf{30}, 165007 (2013). \url{https://doi.org/10.1088/0264-9381/30/16/165007}
\bibitem{Bera2015} S. Bera, R. Mohan, and T. P. Singh, Phys. Rev. D \textbf{92}, 025054 (2015). \url{https://doi.org/10.1103/PhysRevD.92.025054}
\bibitem{Giovannini2024} M. Giovannini, ``Relic gravitons and non-stationary processes,'' JCAP \textbf{11} (2024) 027. \url{https://doi.org/10.1088/1475-7516/2024/11/027}
\bibitem{Calzetta2008} E. Calzetta and B. L. Hu, \textit{Nonequilibrium Quantum Field Theory} (Cambridge University Press, Cambridge, 2008).
\bibitem{Hu2004} B. L. Hu and E. Verdaguer, Living Rev. Relativ. \textbf{7}, 3 (2004). \url{https://doi.org/10.12942/lrr-2004-3}
\bibitem{Martin1999} R. Mart\'in and E. Verdaguer, Phys. Rev. D \textbf{60}, 084008 (1999). \url{https://doi.org/10.1103/PhysRevD.60.084008}
\bibitem{Phillips2001} N. G. Phillips and B. L. Hu, Phys. Rev. D \textbf{63}, 104001 (2001). \url{https://doi.org/10.1103/PhysRevD.63.104001}
\bibitem{Breuer2002} H.-P. Breuer and F. Petruccione, \textit{The Theory of Open Quantum Systems} (Oxford University Press, Oxford, 2002).
\bibitem{Lindblad1976} G. Lindblad, Commun. Math. Phys. \textbf{48}, 119 (1976). \url{https://doi.org/10.1007/BF01608499}
\bibitem{Gorini1976} V. Gorini, A. Kossakowski, and E. C. G. Sudarshan, J. Math. Phys. \textbf{17}, 821 (1976). \url{https://doi.org/10.1063/1.522979}
\bibitem{Oppenheim2023} J. Oppenheim, Phys. Rev. X \textbf{13}, 041040 (2023). \url{https://doi.org/10.1103/PhysRevX.13.041040}

\end{thebibliography}
\end{document}